\documentclass[journal,a4paper,onecolumn]{IEEEtran}
\usepackage{graphicx}
\usepackage{amssymb,amsbsy}
\usepackage{epstopdf}
\usepackage{amsmath,amsthm}
\usepackage{enumerate}
\usepackage{eufrak}
\usepackage{cite}
\usepackage{mathcomp}
\usepackage{supertabular}
\usepackage{longtable}
\usepackage{stmaryrd}
\usepackage{url}
\usepackage{color}
\usepackage{rotating}
\usepackage{float}
\usepackage[utf8]{inputenc}

\usepackage{array}
\usepackage{mathtools}
\usepackage{multicol}
\usepackage{amsmath}
\usepackage{algorithm}
\usepackage{algorithmic}
\usepackage{amssymb}
\usepackage{xcolor}

\usepackage[mathcal]{eucal}

\usepackage[all]{xy}
\entrymodifiers={++[o][F-]}

\usepackage{setspace}
\theoremstyle{definition}

\newtheorem{prop}{Proposition}[section]
\newtheorem{thm}[prop]{Theorem}
\newtheorem{cor}[prop]{Corollary}
\newtheorem{lem}[prop]{Lemma}

\newtheorem{defn}{Definition}[section]
\newtheorem{exa}{Example}[section]

\newtheorem{rem}{Remark}

\begin{document}

\newcommand{\vA}{{\bf A}}
\newcommand{\vAtilde}{\widetilde{\bf A}}

\newcommand{\vB}{{\bf B}}
\newcommand{\vBtilde}{\widetilde{\bf B}}

\newcommand{\vC}{{\bf C}}
\newcommand{\vD}{{\bf D}}
\newcommand{\vH}{{\bf H}}
\newcommand{\vI}{{\bf I}}

\newcommand{\vY}{{\bf Y}}
\newcommand{\vZ}{{\bf Z}}

\newcommand{\vJ}{{\bf J}}

\newcommand{\vM}{{\bf M}}
\newcommand{\vN}{{\bf N}}
\newcommand{\vU}{{\bf U}}
\newcommand{\vV}{{\bf V}}
\newcommand{\vT}{{\bf T}}
\newcommand{\vR}{{\bf R}}
\newcommand{\vS}{{\bf S}}

\newcommand{\va}{{\bf a}}
\newcommand{\vb}{{\bf b}}
\newcommand{\vc}{{\bf c}}

\newcommand{\ve}{{\bf e}}
\newcommand{\vh}{{\bf h}}
\newcommand{\vp}{{\bf p}}
\newcommand{\vs}{{\bf s}}

\newcommand{\vu}{{\bf u}}
\newcommand{\vv}{{\bf v}}
\newcommand{\vw}{{\bf w}}
\newcommand{\vx}{{\bf x}}
\newcommand{\vr}{{\bf r}}
\newcommand{\vhx}{{\widehat{\bf x}}}
\newcommand{\vtx}{{\widetilde{\bf x}}}
\newcommand{\vhp}{{\widehat{\bf p}}}
\newcommand{\vy}{{\bf y}}
\newcommand{\vz}{{\bf z}}

\newcommand{\tx}{{\widetilde{x}}}

\newcommand{\vj}{{\bf j}}
\newcommand{\vzero}{{\bf 0}}
\newcommand{\vone}{{\bf 1}}
\newcommand{\vbeta}{{\boldsymbol \beta}}
\newcommand{\vchi}{{\boldsymbol \chi}}

\newcommand{\tA}{\textrm A}
\newcommand{\tB}{\textrm B}
\newcommand{\A}{\mathcal A}
\newcommand{\B}{\mathcal B}
\newcommand{\C}{\mathcal C}
\newcommand{\D}{\mathcal D}
\newcommand{\E}{\mathcal E}
\newcommand{\F}{\mathcal F}
\newcommand{\G}{\mathcal G}
\newcommand{\M}{\mathcal M}
\newcommand{\HH}{\mathcal H}
\newcommand{\PP}{\mathcal P}

\newcommand{\Q}{\mathcal Q}
\newcommand{\Qb}{\bar{\mathcal Q}}
\newcommand{\Db}{{\bar{\Delta}}}

\newcommand{\pQ}{{\bf p}\mathcal Q}
\newcommand{\pQb}{{\bf p}\bar{\mathcal Q}}

\newcommand{\LL}{\mathcal L}
\newcommand{\R}{\mathcal R}
\newcommand{\SSS}{\mathcal S}
\newcommand{\U}{\mathcal U}
\newcommand{\V}{\mathcal V}
\newcommand{\Y}{\mathcal Y}
\newcommand{\Z}{\mathcal Z}

\newcommand{\Pg}{{{\mathcal P}_{\rm gram}}}
\newcommand{\Pgint}{{{\mathcal P}^\circ_{\rm gram}}}
\newcommand{\Pgrc}{{{\mathcal P}_{\rm GRC}}}
\newcommand{\Pgrcint}{{{\mathcal P}^\circ_{\rm GRC}}}
\newcommand{\Pint}{{{\mathcal P}^\circ}}
\newcommand{\Ag}{{\bf A}_{\rm gram}}

\newcommand{\CC}{\mathbb C} 
\newcommand{\RR}{\mathbb R}
\newcommand{\ZZ}{\mathbb Z}
\newcommand{\nonneg}{\ZZ_{\ge 0}}
\newcommand{\FF}{\mathbb F}
\newcommand{\KK}{\mathbb K}

\newcommand{\Fnd}{\FF_q^{n^{\otimes d}}}
\newcommand{\Knd}{\KK^{n^{\otimes d}}}

\newcommand{\ceiling}[1]{\left\lceil{#1}\right\rceil}
\newcommand{\floor}[1]{\left\lfloor{#1}\right\rfloor}
\newcommand{\bbracket}[1]{\left\llbracket{#1}\right\rrbracket}

\newcommand{\inprod}[1]{\left\langle{#1}\right \rangle}


\newcommand{\beas}{\begin{eqnarray*}} 
\newcommand{\eeas}{\end{eqnarray*}} 

\newcommand{\bm}[1]{{\mbox{\boldmath $#1$}}} 

\newcommand{\sizeof}[1]{\left\lvert{#1}\right\rvert}
\newcommand{\wt}{{\rm wt}} 
\newcommand{\capty}{\alpha} 

\newcommand{\supp}{{\rm supp}} 
\newcommand{\dg}{d_{\rm gram}} 
\newcommand{\da}{d_{\rm asym}} 
\newcommand{\dist}{{\rm dist}} 
\newcommand{\ssyn}{s_{\rm syn}}
\newcommand{\sseq}{s_{\rm seq}}
\newcommand{\nullplus}{{\rm Null}_{>\vzero}}

\newcommand{\tworow}[2]{\genfrac{}{}{0pt}{}{#1}{#2}}
\newcommand{\qbinom}[2]{\left[ {#1}\atop{#2}\right]_q}

\newcommand{\Lovasz}{Lov\'{a}sz }
\newcommand{\Mobius}{M\"obius}
\newcommand{\etal}{\emph{et al.}}

\newcommand{\todo}[1]{{\color{red} (TODO: #1) }}


\title{Rates of DNA Sequence Profiles for\\Practical Values of Read Lengths}

\author{
Zuling Chang\IEEEauthorrefmark{1}, 
Johan Chrisnata\IEEEauthorrefmark{2}, 
Martianus Frederic Ezerman\IEEEauthorrefmark{2}, and Han Mao Kiah\IEEEauthorrefmark{2}\\
\IEEEauthorblockA{\IEEEauthorrefmark{1}School of Mathematics and Statistics, Zhengzhou University, China\\
\IEEEauthorrefmark{2}School of Physical and Mathematical Sciences, Nanyang Technological University, Singapore\\
		 Emails: ${\rm zuling\_chang}$@zzu.edu.cn, ${\rm\{jchrisnata,\, fredezerman,\, hmkiah\}}$@ntu.edu.sg}  
\thanks{The work of Z.~Chang is supported by the Joint Fund of the National Natural Science Foundation of China under Grant U1304604. Research Grants TL-9014101684-01 and MOE2013-T2-1-041 support M.~F.~Ezerman. }
\thanks{Earlier results of this paper were presented at the 2016 Proceedings of the IEEE ISIT \cite{Chang:2016}.}

}


\maketitle

\begin{abstract}

A recent study by one of the authors has demonstrated the importance of profile vectors in DNA-based data storage.
We provide exact values and lower bounds on the number of profile vectors for finite values of alphabet size $q$, 
read length $\ell$, and word length $n$.
Consequently, we demonstrate that for $q\ge 2$ and  $n\le q^{\ell/2-1}$,
the number of profile vectors is at least $q^{\kappa n}$ with $\kappa$ very close to 1.
In addition to enumeration results, we provide a set of efficient encoding and decoding algorithms 
for each of two particular families of profile vectors.
\end{abstract}


\section{Introduction}


Despite advances in traditional data recording techniques,
the emergence of Big Data platforms and energy conservation issues
impose new challenges to the storage community in terms of identifying 
high volume, nonvolatile, and durable recording media. 
The potential for using macromolecules for ultra-dense storage was recognized as early as in the 1960s.
Among these macromolecules, DNA molecules 
stand out due to their biochemical robustness and high storage capacity.

In the last few decades, the technologies for synthesizing (writing) artificial DNA and 
for massive sequencing (reading) have reached attractive levels of efficiency and accuracy.
Building upon the rapid growth of DNA synthesis and sequencing technologies, 
two laboratories recently outlined architectures for archival DNA-based storage \cite{Church.etal:2012,Goldman.etal:2013}. 
The first architecture achieved a density of 700 TB/gram, while the second approach raised the density to 2.2 PB/gram. 
To further protect against errors, Grass \etal{} later incorporated Reed-Solomon error-correction schemes and 
encapsulated the DNA media in silica \cite{Grass.etal:2015}.
Yazdi \etal{} recently proposed a completely different approach and provided a random access and 
rewritable DNA-based storage system \cite{Yazdi.etal:2015,Yazdi.etal:2015b}.

More recently, to control specialized errors arising from sequencing platforms, two families of codes were introduced 
by Gabrys \etal \cite{Gabrys.etal:2015} and Kiah \etal \cite{Kiah.etal:2015a}. 
The former looked at miniaturized nanopore sequencers such as MinION, 
while the latter focused on errors arising from high-throughput sequencers such as Illumina, 
which arguably is the more mature technology.
The latter forms the basis for this work.
In particular, we examine the concept of {\em DNA profile vectors} introduced by Kiah \etal \cite{Kiah.etal:2015a}.

In this channel model, to store and retrieve information in DNA, 
one starts with a desired information sequence encoded into a sequence or word defined over the nucleotide alphabet 
$\{{\tt A}, {\tt C}, {\tt G}, {\tt T}\}$. 
The \emph{DNA storage channel} models a physical process which takes as its input the sequence of length $n$, and 
synthesizes (writes) it physically into a macromolecule string.
To retrieve the information, the user can proceed using several read technologies. 
The most common sequencing process, implemented by Illumina, 
makes numerous copies of the string or amplifies the string, 
and then fragments all copies of the string into a collection of substrings (reads) 
of approximately the same length $\ell$, so as to produce a large number of overlapping {\em reads}. 
Since the concentration of all (not necessarily) distinct substrings within the mix is usually assumed to be uniform, one may normalize the concentration of all subsequences by the concentration of the least abundant substring.
As a result, one actually observes substring concentrations reflecting the frequency of the substrings in \emph{one copy} of the original string.
Therefore, we model the output of the channel as an 
\emph{unordered subset of reads}. This set may be summarized by its multiplicity vector, 
which we call the \emph{output profile vector}.

We assume an errorless channel and 
observe that it is possible for different words or strings to have an identical profile vector.
Hence, even without errors, the channel may be unable to distinguish between certain pairs of words.
Our task is then to enumerate all distinct profile vectors for fixed values of $n$ and $\ell$ over a $q$-ary alphabet.

In the case of arbitrary $\ell$-substrings, the problem of enumerating all valid profile vectors 
was addressed by Jacquet \etal{} in the context of ``Markov types'' \cite{Jacquet.etal:2012}. 
Kiah \etal{} then extended the enumeration results to profiles with specific $\ell$-substring constraints
so as to address certain considerations in DNA sequence design \cite{Kiah.etal:2015a}.
In particular, for fixed values of $q$ and $\ell$, 
the number of profile vectors is known to be $\Theta\left(n^{q^\ell-q^{\ell-1}}\right)$.

However, determining the coefficient for the dominating term $n^{q^\ell-q^{\ell-1}}$ is a computationally difficult task. It has been determined for only very small values of $q$ and $\ell$ in \cite{Jacquet.etal:2012, Kiah.etal:2015a}.
Furthermore, it is unclear how accurate the asymptotic estimate $\Theta\left(n^{q^\ell-q^{\ell-1}}\right)$ is for practical values of $n$.
Indeed, most current DNA storage systems do not use string lengths $n$ 
exceeding several thousands nucleotides (nts) due to the high cost of synthesis. 
On the other hand, current sequencing systems have read length $\ell$ between $100$ to $1500$ nts.

This paper adopts a different approach and looks for lower bounds 
for the number of profile vectors given moderate values of $q$, $\ell$, and $n$.
Surprisingly, for fixed $q\ge 2$ and moderately large values $n\le q^{\ell/2-1}$,
the number of profile vectors is at least $q^{\kappa n}$ with $\kappa$ very close to 1.
As an example, when $q=4$ (the number of DNA nucleotide bases) and $\ell=100$ (a practical read length),
our results show that there are at least $4^{0.99n}$ distinct  $100$-gram profile vectors for $1000\le n\le 10^6$.
In other words, for practical values of read and word lengths, 
we are able to obtain a set of distinct profile vectors with rates {\em close to one}.
In addition to enumeration results, we propose a set of linear-time encoding and decoding algorithms 
for each of two particular families of profile vectors.

\section{Preliminaries and Main Results}

Let $\llbracket q\rrbracket$ denote the set of integers $\{0,1,\ldots, q-1\}$ and 
$[i,j]$ denote the set of integers $\{i,i+1,\ldots, j\}$.
Consider a word $\vx=x_1x_2\cdots x_n$ of length $n$ over $\llbracket q\rrbracket$. 
For $1\le i<j\le n$, we denote  
the entry $x_i$ by $\vx[i]$, the {\em substring} $x_ix_{i+1}\cdots x_j$ of length $(j-i+1)$ by $\vx[i,j]$, 
and the length of $\vx$ by $|\vx|$.

For $\ell\le n$ and $1\le i\le n-\ell+1$, we also call the substring $\vx[i,i+\ell-1]$ an {\em $\ell$-gram} of $\vx$.
For $\vz\in\bbracket{q}^\ell$, let $p(\vx,\vz)$ denote the number of occurrences of $\vz$ as an $\ell$-gram of $\vx$.
Let $\vp(\vx,\ell)\triangleq \Big(p(\vx,\vz)\Big)_{\vz\in\bbracket{q}^\ell}$ be the ($\ell$-gram) {\em profile vector} of length $q^\ell$, indexed by all words of $\llbracket q\rrbracket^{\ell}$ ordered lexicographically. 
Let $\F(\vx,\ell)$ be the set of $\ell$-grams of $\vx$. 
In other words, $\F(\vx,\ell)$ is the support for the vector $\vp(\vx,\ell)$.

\begin{exa}\label{exa:simple} Let $q=2$, $n=5$ and $\ell=2$. 
Then $p(10001,01)=p(10001,10)=1$, while $p(10001,00)=2$.
So, $\vp(10001,2)=(2,1,1,0)$ and $\F(10001,2)=\{00,01,10\}$.

Consider the words $00010$ and $00101$. Then 
$\vp(10001,2)=\vp(00010,2)$ while $\F(10001,2)=\F(00010,2)=\F(00101,2)$.
\end{exa}

As illustrated by Example \ref{exa:simple}, different words may have the same profile vector.
We define a relation on $\bbracket{q}^n$ where $\vx\sim \vx'$ if and only if
$\vp(\vx,\ell)=\vp(\vx',\ell)$. 
It can be shown that $\sim$ is an equivalence relation and we denote the number of equivalence classes by $P_q(n,\ell)$.
We further define  the 
 {\em rate of profile vectors} to be $R_q(n,\ell)=\log_q P_q(n,\ell)/n$.

The asymptotic growth of $P_q(n,\ell)$ as a function of $n$ is given as below. 

\begin{thm}[{Jacquet \etal{} \cite{Jacquet.etal:2012}, Kiah \etal{}\cite{Kiah.etal:2015a}}]
Fix $q\ge 2$ and $\ell$. Then \[P_q(n,\ell)=\Theta\left(n^{q^\ell-q^{\ell-1}}\right).\]
Hence, $\lim_{n\to\infty} R_q(n,\ell)=0$.
\end{thm}

Our main contribution is the following set of exact values and lower bounds for $P_q(n,\ell)$ 
for finite values of $n$, $q$ and $\ell$.

\begin{thm}\label{thm:main}
Fix $q\ge 2$. Let $\mu$ be the {\Mobius} function.

\begin{enumerate}[(i)]
\item If $\ell\le n< 2\ell$, then
\begin{equation}\label{nsmall}
P_q(n,\ell)=q^n-\sum_{r \mid n-\ell+1}\sum_{t \mid r}\left(\frac{r-1}{r}\right)\mu\left(\frac rt\right) q^{t}>q^{n-1}(q-1).
\end{equation}

\item If $n=q^{a-1}\ell$ where $4\le 2a\le \ell$, then
\begin{equation}\label{addressing}
P_q(n,\ell)\ge (q-1)^{q^{a-1}(\ell-a)}.
\end{equation}

\item If $n\ge \ell$, 

\begin{equation}\label{debruijn}
P_q(n,\ell)> \frac1n \left(\sum_{t \mid n}\mu\left(\frac nt\right) q^{t}-\binom{n}{2}q^{n-\ell+1}\right).
\end{equation}
Suppose further that $\ell\ge 2\log_q n+2$ and $n\ge 8$. Then $P_q(n,\ell)>q^{n-1}/n$.
\end{enumerate}
\end{thm}

We prove Equations \eqref{nsmall}, \eqref{addressing}, and \eqref{debruijn} 
in Sections \ref{sec:nsmall}, \ref{sec:addressing}, and \ref{sec:debruijn},  respectively.

\begin{figure*}[!t]
\begin{tabular}{ll}
(a) The rates $R_4(n,20)$ for $20\le n\le 10^8$ &
(b) The rates $R_4(n,100)$ for $100\le n\le 10^8$\\

\includegraphics[height=65mm]{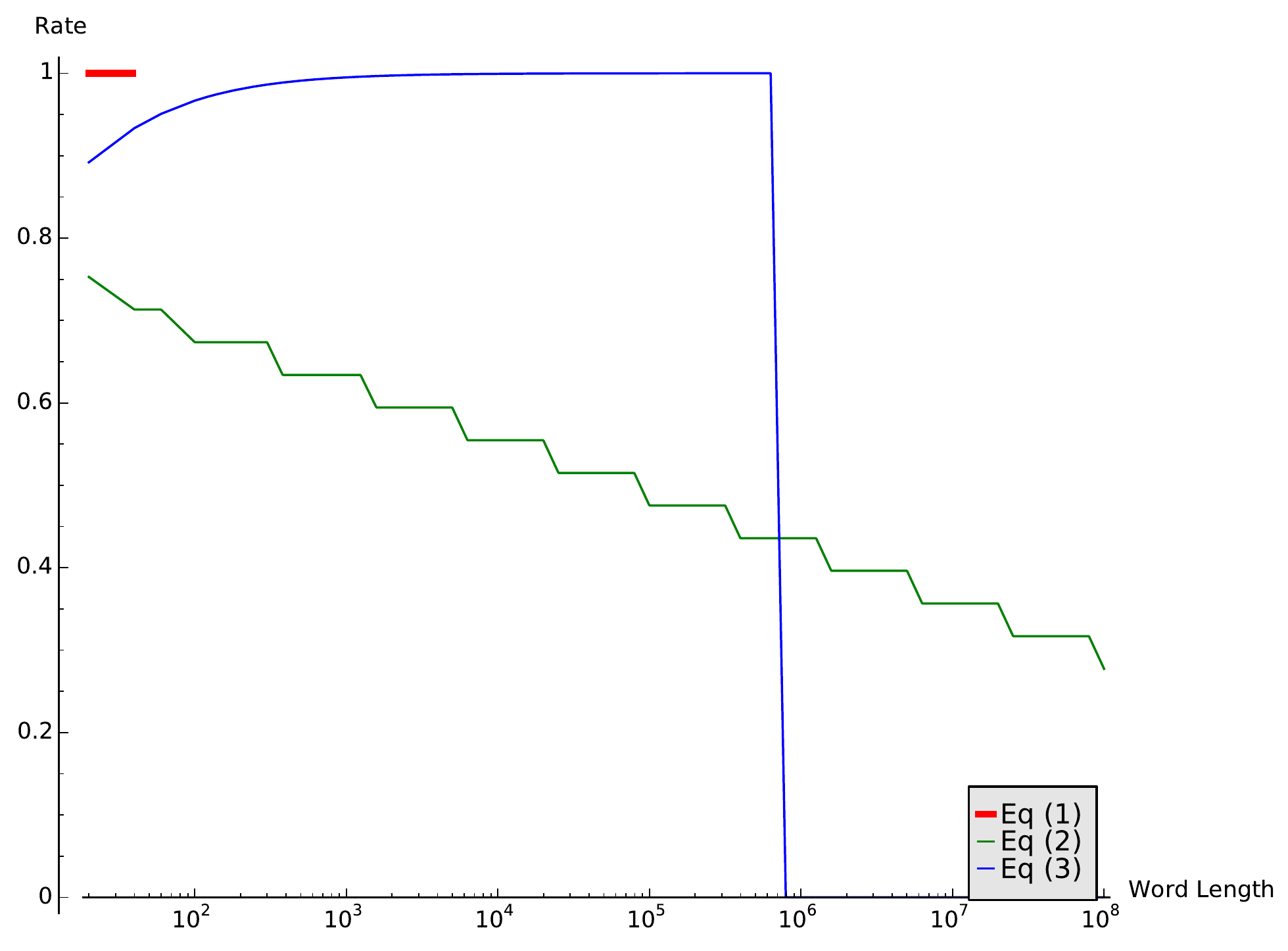}
&
\includegraphics[height=65mm]{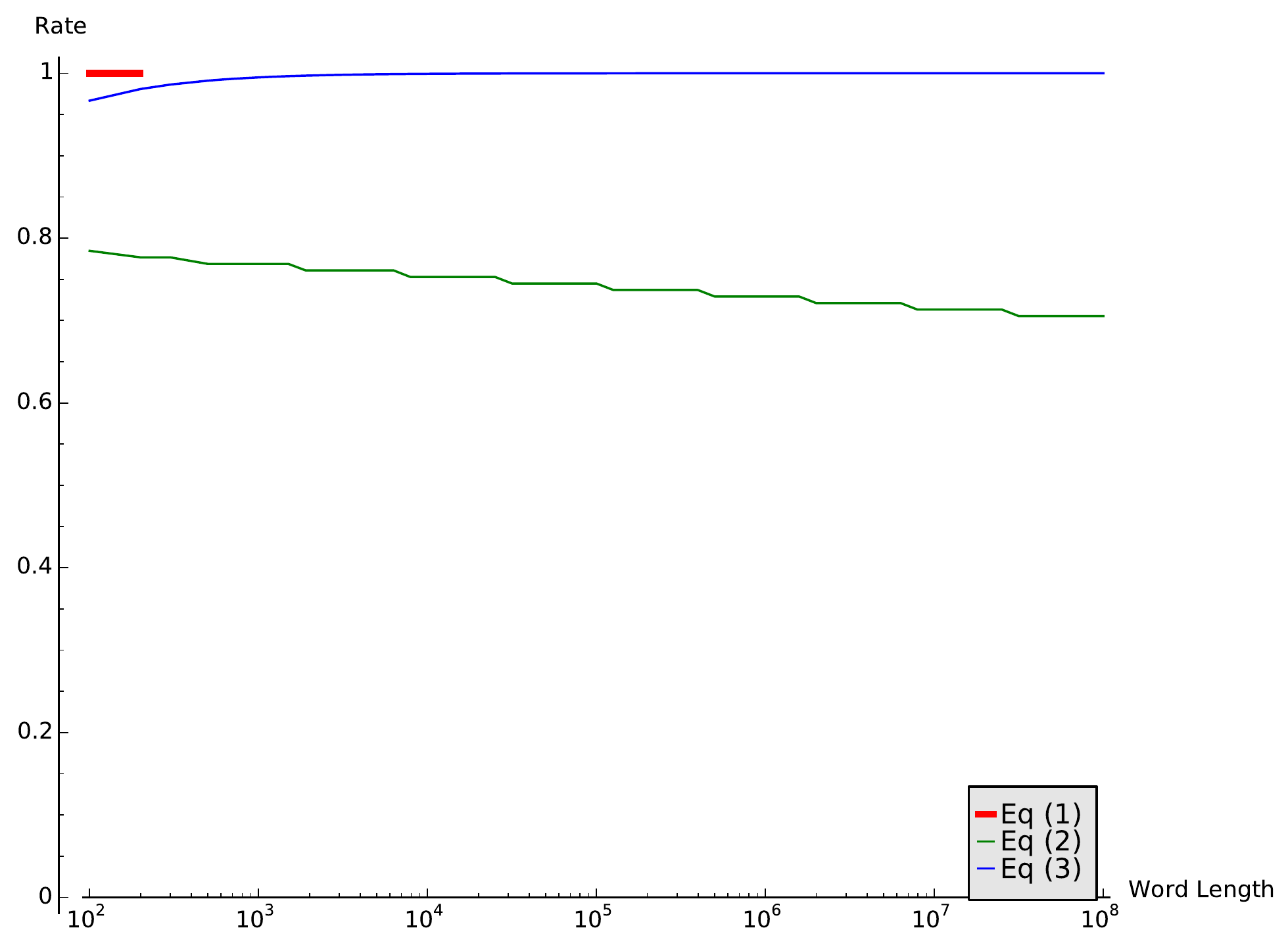}
\end{tabular}
\caption{Rate of profile vectors for fixed values of $\ell$.}
\label{fig:ellfixed}
\end{figure*}

\begin{figure*}[!t]
\begin{tabular}{ll}
(a) The rates $R_4(1000,\ell)$ for $1\le \ell\le 1000$ &
(b) The rates $R_4(10^6,\ell)$ for $1\le \ell \le 10^6$ \\
\includegraphics[height=65mm]{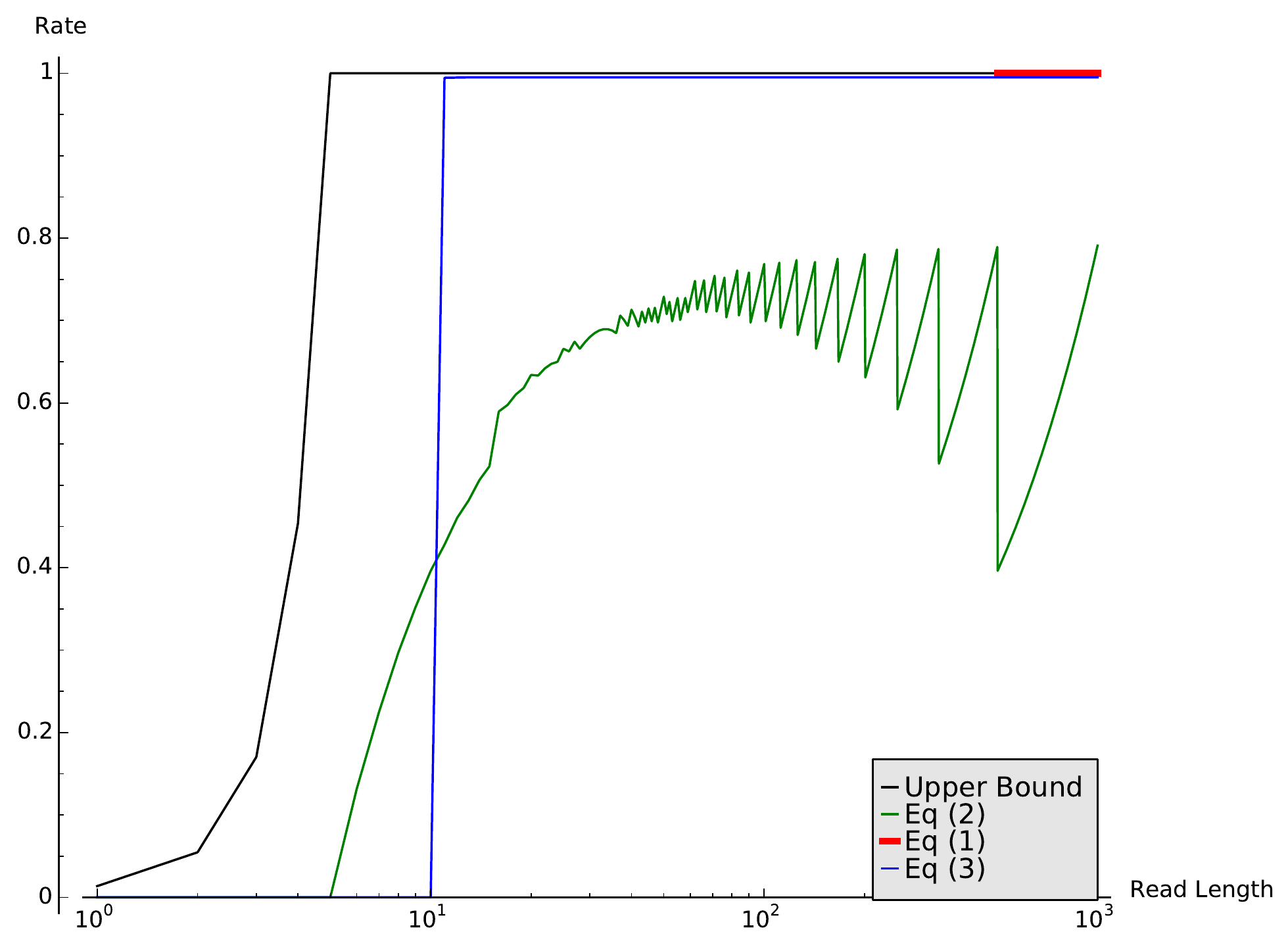}
& \includegraphics[height=65mm]{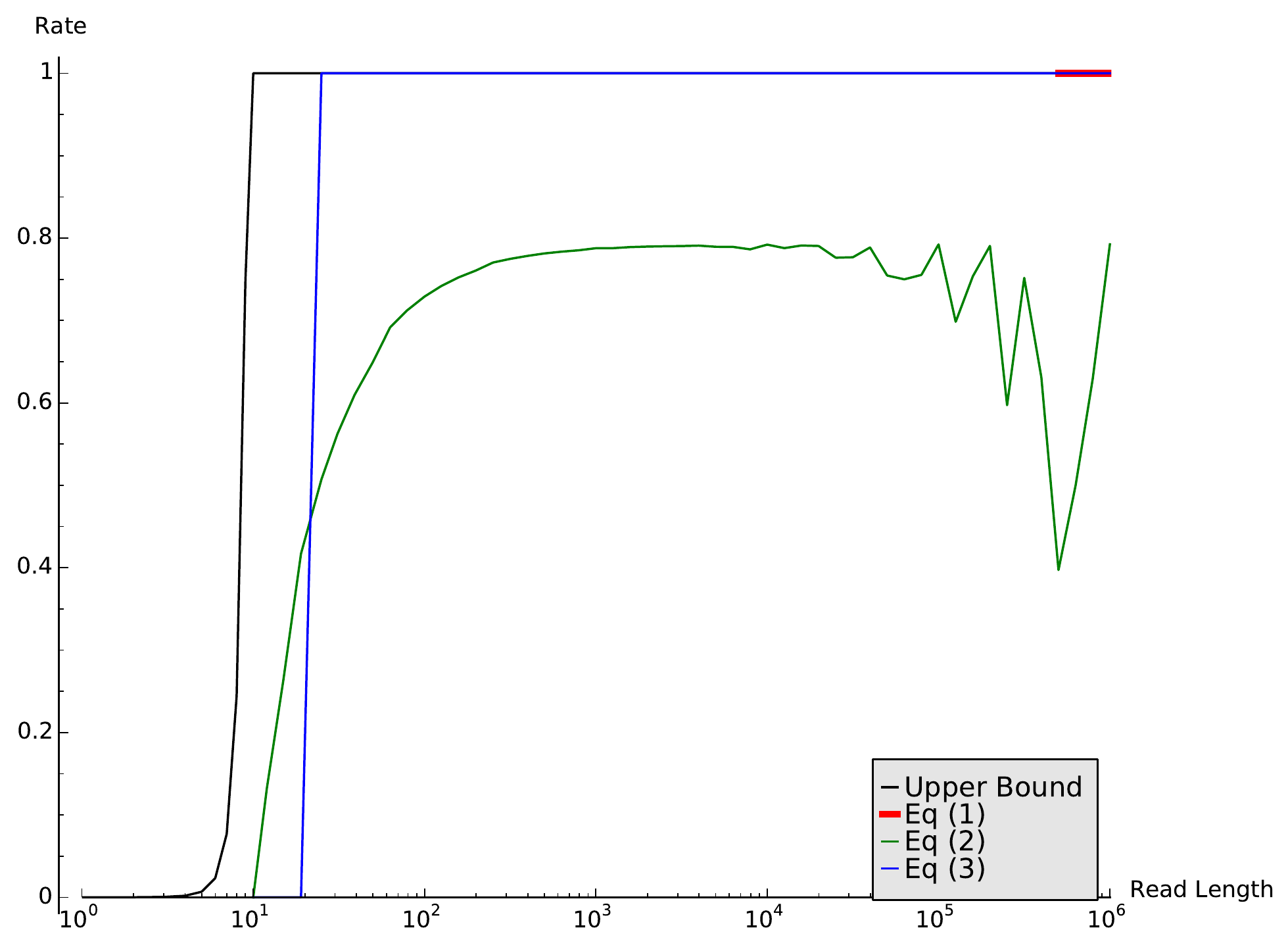}
\end{tabular}
\caption{Rate of profile vectors for fixed values of $n$.}
\label{fig:nfixed}
\end{figure*}

\begin{rem}
We compare the rates provided in Theorem \ref{thm:main} with $q=4$ in Figures \ref{fig:ellfixed} and \ref{fig:nfixed}.
\begin{enumerate}[(i)]
\item Figure \ref{fig:ellfixed}(b) shows that for a practical read length $\ell=100$ and word lengths $n\le 10^8$,  
the rates of the profile vectors is very close to one. 
In fact, computations show that $R_4(n,\ell)\ge 0.99$ for $1000\le n\le 10^6$.
Even for a shorter read length $\ell=20$, Figure \ref{fig:ellfixed}(a) illustrates that the rates are close to one for word lengths $n\le 10^5$. 
Therefore, for practical values of read and word lengths, 
we obtain a set of distinct profile vectors with rates close to one.

\item In Figure \ref{fig:nfixed}, we plot an upper bound for $R_q(n,\ell)$, given by
\begin{equation}\label{upper}
R_q(n,\ell)\le \frac1n \log_q \binom{n-\ell+q^\ell}{q^\ell-1}.
\end{equation}
Here, the inequality follows from the fact that a profile vector is an integer-valued vector of length $q^\ell$, 
whose entries sum to $n-\ell+1$. 
Such vectors are also known as {\em weak $q^\ell$-compositions of $(n-\ell+1)$} and 
are given by $\binom{n-\ell+q^\ell}{q^\ell-1}$ (see for example Stanley \cite[Section 1.2]{Stanley:2011}),
and hence, the inequality follows.

Figure \ref{fig:nfixed} illustrates that if we fix the word length $n$, 
we have $R_q(n,\ell)\approx 1$ for $\ell\ge 2\log_q n$ and $R_q(n,\ell)\approx 0$ for $\ell\le \log_q n$.
Therefore, it remains open to determine $R_q(n,\ell)$ for $\log_q n\le \ell\le 2\log_q n$.
We state this observation formally in Corollary \ref{cor:rates} and  provide an asymptotic analysis for the rates of profile vectors.

\item 
Figures \ref{fig:ellfixed} and \ref{fig:nfixed} shows that, for $n\le q^{\ell/2 -1}$ or $\ell\ge 2\log_q n+2$, Equation \eqref{debruijn} provides 
a significantly better lower bound than Equation \eqref{addressing}. 
However, we are only able to demonstrate a set of efficient encoding and decoding algorithms 
for the two families of profile vectors associated with Equations \eqref{nsmall} and \eqref{addressing}.

Furthermore, we provide an efficient sequence reconstruction algorithm for the family of profile vectors associated with Equation \eqref{addressing}.

\end{enumerate}
\end{rem}

Let $n$ be a function of $\ell$, or $n=n(\ell)$ such that $n(\ell)$ increases with $\ell$. 
We then define the 
{\em asymptotic rate of profile vectors with respect to $n$} via the equation
\begin{equation}\label{def:cap}
\capty(n,q)\triangleq \limsup_{\ell\to\infty} R_q(n,\ell). 
\end{equation}
Suppose that $\ell$ is a system parameter determined by current sequencing technology.
Then $n=n(\ell)$ determines how long we can set our codewords so that 
the information rate of the DNA storage channel remains as $\capty(n,q)$.
From Theorem \ref{thm:main}, we derive the following result on the asymptotic rates.
The detailed proof is given in Section \ref{sec:debruijn}.

\begin{cor}[Asymptotic rates] \label{cor:rates} Fix $q\ge 2$. 
\begin{enumerate}
\item Suppose that $\ell\le n(\ell) \le q^{\ell/2-1} $ for all $\ell$.
Then $\alpha(n,q)=1$.
\item Let $\epsilon>0$. Suppose that $n(\ell) \ge q^{(1+\epsilon)\ell} $ for all $\ell$.
Then $\alpha(n,q)=0$.
\end{enumerate}
\end{cor}

%
%





\section{Exact Enumeration of Profile Vectors}\label{sec:nsmall}

We extend the methods of Tan and Shallit \cite{tan2013sets}, where the number of possible $\F(\vx,\ell)$ 
was determined for $\ell\le n<2\ell$. Specifically, we compute $P_q(n,\ell)$ for $\ell\le n<2\ell$.
Our strategy is to first define an equivalence relation using the notions of root conjugates
so that the number of equivalence classes yields $P_q(n,\ell)$. 
We then compute this number using standard combinatorial methods.

\begin{defn}
Let $\vx$ be a $q$-ary word. 
A \textit{period} of $\vx$ is a positive integer $r$ such that $\vx$ can be {\em factorized} as
\[
\vx = \underbrace{\vu\,\vu\,\cdots\,\vu}_{k \text{ times}}\vu'\text{, with } |\vu|=r , \vu' 
\text{ a prefix of }\vu \text{, and } k\geq 1.
\]
Let $\pi(\vx)$ denote the {\em minimum period} of $\vx$. 
The \textit{root} of $\vx$ is given by $\vh(\vx)=\vx[1,\pi(\vx)]$, which is the prefix of $\vx$ with length $\pi(\vx).$ 
Two words $\vx$ and $\vx'$ are said to be  \textit{root-conjugate} if $\vh(\vx)=\vu\vv$ and $\vh(\vx')=\vv\vu$ 
for some words $\vu$ and $\vv$, or $\vh(\vx)$ is a {\em rotation} of $\vh(\vx')$.
\end{defn}

\begin{exa} 10010010 has minimum period three and its root is 100. 
Also, 01001001 has minimum period three and its root is 010.
Therefore, 10010010 and 01001001 are root-conjugates. 
\end{exa}

Observe that two words that are root-conjugates necessarily have the same minimum period
and it can be shown that being root-conjugates form an equivalence relation.
In addition, we have the following technical lemma. 
\begin{lem}\label{lem:root}
Let $\vx$ be a word of length $n$ with $\pi(\vx)\le n-\ell+1\le \ell$.
Then, for $1\le i<j\le \pi(\vx)$, we have $\vx[i,i+\ell-1]\ne \vx[j,j+\ell-1]$.
\end{lem}

\begin{proof}
Suppose that $\vx[i,i+\ell-1]=\vx[j,j+\ell-1]$. Setting $k=j-i$, we have $\vx[s] = \vx[s+k]\,$ for $i \leq s \leq i + \ell -1.$ 
Since $\pi(\vx) \leq n-\ell +1 \leq \ell$, then $\vx[s] = \vx[s+k]$ for $1\leq s \leq \pi(\vx).$ 
Therefore, $\vx[s] = \vx[s+d]$ for $1 \leq s \leq \pi(\vx)$ where $d = \gcd(k,\pi(\vx))\le k=j-i<\pi(\vx)$.
In other words, $\vx$ has a period $d< \pi(\vx)$, 
contradicting the minimality of $\pi(\vx)$.
\end{proof}

Tan and Shallit proved the following result that characterized $\F(\vx,\ell)$ when $|\vx|<2\ell$.

\begin{lem}[{Tan and Shallit \cite[Theorem 15]{tan2013sets}}]\label{lem:ts}
Suppose that $\ell \leq n < 2 \ell$ and $\vx$ and $\vx'$ are distinct $q$-ary words of length $n$.
Then $\F(\vx,\ell) = \F(\vx',\ell)$ if and only if 
$\vx,\vx'$ are root-conjugates with $\pi(\vx)\leq n-\ell +1$.
\end{lem}

%
%

Using Lemma \ref{lem:root}, we extend Lemma \ref{lem:ts} to characterize the profile vectors when $n<2\ell$.

\begin{thm}\label{thm:lyndon} Let $\vx$ and $\vx'$ be distinct $q$-ary words of length $n$.
If $\vx,\vx'$ are root-conjugates with  $\pi(\vx) \mid n-\ell +1$, then $\vp(\vx,\ell)=\vp(\vx',\ell)$.
Conversely, if $\ell \leq n < 2 \ell$ and $\vp(\vx,\ell)=\vp(\vx',\ell)$, then
$\vx,\vx'$ are root-conjugates with  $\pi(\vx) \mid n-\ell +1$.
\end{thm}

\begin{proof}
Suppose that $\vx$ and $\vx'$ are root-conjugates with $\pi(\vx)=r$ and $n-\ell +1=rs$ for some $s$.
Then it can be verified that $\F(\vx,\ell)=\F(\vx',\ell)=\{\vx[i,i+\ell-1]: 1\le  i\le r\}$
and $p(\vx,\vz)=p(\vx',\vz)=s$ for all $\vz\in \F(\vx,\ell)$.
Therefore, $\vp(\vx,\ell)=\vp(\vx',\ell)$.

Conversely, let $\vp(\vx,\ell)=\vp(\vx',\ell)$. Then $\F(\vx,\ell)=\F(\vx',\ell)$. By Lemma \ref{lem:ts}, 
$\vx,\vx'$ are root-conjugates with $\pi(\vx)\leq n-\ell +1$. Let $r=\pi(\vx)$. It remains to show that $r \mid n-\ell+1$.

Suppose otherwise and let $n-\ell+1=rs+t$ with $1\le t\le r-1$. Let the roots of $\vx$ and $\vx'$ be $\vu\vv$ and $\vv\vu$, respectively. Therefore, we can write $\vx$ and $\vx'$ as
\[
\vx = \overbrace{\underbrace{\vu\vv}_{r}\underbrace{\vu\vv}_{r}\cdots\underbrace{\vu\vv}_{r}}^{s \text{ times}}\underbrace{\vw}_{t+\ell-1}\mbox{\quad and  \quad}
\vx' = \overbrace{\underbrace{\vv\vu}_{r}\underbrace{\vv\vu}_{r}\cdots\underbrace{\vv\vu}_{r}}^{s \text{ times}}
\underbrace{\vw'}_{t+\ell-1}.
\]

We have the following cases.
\begin{enumerate}[(i)]
\item If $1\le t<|\vu|$, let $\vz'$ be the $\ell$-length prefix of $\vx'$. 
Since $|\vw'|=t+\ell-1\ge \ell$ and $\vz'$ is a prefix of $\vw'$, we have $p(\vx',\vz')\ge s+1$.
On the other hand, from Lemma \ref{lem:root}, 
the $\ell$-gram of $\vz'$ can only appear after the first $|\vu|$ coordinates of $\vw$. 
However, $|\vw|-|\vu|< t+(\ell-1)-t<\ell$, and so, there is no occurrence of $\vz'$ as an $\ell$-gram of $\vw$.
Therefore, $p(\vx,\vz')=s<p(\vx',\vz')$, contradicting the assumption that $\vp(\vx,\ell)=\vp(\vx',\ell)$.
\item If $|\vu|\le t\le r-1$, let $\vz=\vx[|\vu|,|\vu|+\ell-1]$.
Since $|\vw|=t+\ell-1\ge |\vu|+\ell-1$, we have $p(\vx,\vz)\ge s+1$.
With the same considerations as before, we check that there is no occurence of $\vz$ as an $\ell$-gram of $\vw'$,
So, $p(\vx',\vz)=s<p(\vx,\vz)$, a contradiction.
\end{enumerate}

Therefore, we conclude $t=0$ or $r \mid n-\ell+1$, as desired.
\end{proof}

Hence, for $\ell\le n<2\ell$, we have $\vx\sim\vx'$ if and only if
$\vx$ and $\vx'$ are root-conjugates with  $\pi(\vx) \mid n-\ell +1$.
We compute the number of equivalence classes using this characterization.

A word is said to be {\em aperiodic} if it is not equal to any of its nontrivial rotations.
An aperiodic word of length $r$ is said to be {\em Lyndon} if 
it is the lexicographically least word amongst all of its $r$ rotations.
The number of Lyndon words  \cite{Lyndon:1954} of length $r$ is given by 
\begin{equation}\label{eq:lyndon}
L_q(r)=\frac1r\sum_{t \mid r}\mu\left(\frac {r}{t}\right)q^t.
\end{equation}

For any integer $r \mid n-\ell+1$ and any word $\vx$, 
if $\pi(\vx)=r$ and $\vh(\vx)$ is its root,
then $\vh(\vx)$ is aperiodic and is a rotation of some Lyndon word $\vu(\vx)$.
Let $\vu(\vx)$ be the representative of the equivalence class of $\vx$.
Since there are $r$ rotations of $\vu(\vx)$, 
there are $r$ words in the equivalence class of $\vx$.
Therefore, the number of equivalence classes is
\[q^n-\sum_{r \mid n-\ell+1}(r-1)L_q(r),\]
and, consequently, we obtain \eqref{nsmall}.

\begin{exa}
Let $n=5$, $\ell=4$, and $q=2$. 
Consider the words $\vx=10101$ and $\vx'=01010$, which are root-conjugates with minimum period two.
Since $2|n-\ell+1$, it follows from Theorem \ref{thm:lyndon} that $\vx$ and $\vx'$ have the same profile vector.

Conversely, if there are two distinct words $\vx$ and $\vx'$ such that $\vx\sim\vx'$, then Theorem \ref{thm:lyndon} 
states that the minimum period of $\vx$ divides two. 
It is then not difficult to argue that the pair of words $\vx$ and $\vx'$ must be $10101$ and $01010$.
Therefore, the number of distinct profile vectors $P_2(5,4)$ is 31. More generally, in the case $\ell=n-1$, 
\eqref{nsmall} reduces to $P_q(n,\ell)=q^n-\binom{q}{2}$.
\end{exa}

From Theorem \ref{thm:lyndon},
if $\vx$ and $\vx'$ are root-conjugates with  $\pi(\vx) \mid n-\ell +1$,
we have $\vp(\vx,\ell)=\vp(\vx',\ell)$ for {\em all} values of $n$. 
In other words, the number of equivalence classes computed above provides an upper bound 
for the number of profile vectors. Formally, we have the following corollary.

\begin{cor} For $n\ge 2\ell$,
\[P_q(n,\ell)\le q^n-\sum_{r \mid n-\ell+1}\sum_{t \mid r}\left(\frac{r-1}{r}\right)\mu\left(\frac rt\right) q^{t}.\]
\end{cor}

Next, we assume $n<2\ell$ and provide efficient methods to encode and decode messages into 
$q$-ary words of length $n$ with distinct $\ell$-gram profile vectors.
To do so, we make use of the following simple observation from Theorem \ref{thm:lyndon}.

\begin{lem}\label{lem:encode-lyndon}
Let $n<2\ell$ and $\vx$ be a $q$-ary word of length $n$ such that $\vx[\ell-1]\ne\vx[n]$.
If $\vx\sim\vx'$, then $\vx=\vx'$.
\end{lem}

\begin{proof} Suppose otherwise that $\vx'$ and $\vx$ are distinct. 
Then Theorem \ref{thm:lyndon} implies that $\pi(\vx)|n-\ell+1$.
In other words, $n-\ell+1$ is a period of $\vx$ and 
hence $\vx[n]=\vx[n-(n-\ell+1)]=\vx[\ell-1]$, yielding a contradiction.
\end{proof}

Lemma \ref{lem:encode-lyndon} then motivates the 
encoding and decoding methods presented as Algorithms \ref{encode1} and \ref{decode1}, respectively.
Define $\C$ to be the image of ${\tt encode}_1$. 

\begin{exa}Set $q=2$, $n=5$, and $\ell=4$.

Suppose that we encode $\vc=00001\in \bbracket{2}^4\times\{1\}$.
Applying Algorithm \ref{encode1}, since $\vc[3]\ne\vc[5]$, then $\vx={\tt encode}_1(\vc)=\vc$.

On the other hand, suppose that we encode $\vc=00101\in \bbracket{2}^4\times\{1\}$.
Applying Algorithm \ref{encode1}, since $\vc[3]=\vc[5]$, then $\vx=00100$ by setting the last bit of $\vc$ to zero.

We compute ${\tt encode}_1(\vc)$ for all $\vc\in\bbracket{2}^4\times\{1\}$ to obtain the code $\C$.
\begin{align*}
{\tt encode}_1(00001) & =00001, &
{\tt encode}_1(01001) & =01001, &
{\tt encode}_1(10001) & =10001, &
{\tt encode}_1(11001) & =11001, \\
{\tt encode}_1(00011) & =00011, &
{\tt encode}_1(01011) & =01011, &
{\tt encode}_1(10011) & =10011, &
{\tt encode}_1(11011) & =11011, \\
{\tt encode}_1(00101) & =00100, &
{\tt encode}_1(01101) & =01100, &
{\tt encode}_1(10101) & =10100, &
{\tt encode}_1(11101) & =11100, \\
{\tt encode}_1(00111) & =00110, &
{\tt encode}_1(01111) & =01110, &
{\tt encode}_1(10111) & =10110, &
{\tt encode}_1(11111) & =11110.
\end{align*}
Each word in $\C$ has its third coordinate different from its last coordinate, and 
no two words in $\C$ share the same profile vector.
\end{exa}

We summarize our observations in the following proposition.

\begin{prop}\label{prop:constr}Let $n<2\ell$. Consider the maps ${\tt encode}_1$ and ${\tt decode_1}$ defined by Algorithms \ref{encode1} and \ref{decode1} and the code $\C$. Then $\vp(\vx,\ell)\ne \vp(\vx',\ell)$  for any two distinct words $\vx, \vx'\in\C$ and ${\tt decode}_1\circ{\tt encode}_1(\vc)=\vc$ for all $\vc \in \bbracket{q}^{n-1}\times \{1,2,\ldots,q-1\}$.
Hence, ${\tt encode}_1$ is injective and $|\C|=q^{n-1}(q-1)$.
Furthermore, ${\tt decode}_1$ and ${\tt encode}_1$ computes their respective strings in $O(n)$ time.
\end{prop}

Therefore, for $n<2\ell$, we have $P_q(n,\ell)\ge q^{n-1}(q-1)$. 
Now, observe that 
Algorithm \ref{encode1} encodes $(n-1)\log_2 q+\log_2(q-1)$ bits of information,
while the set of all $q$-ary words of length $n$ has the capacity to encode $n\log_2 q$ bits of information.
Hence, by imposing the constraint that the words have distinct profile vectors, 
we only lose $\log_2 q-\log_2(q-1)\le 1$ bit of information.

%
	
	\begin{algorithm}[!t]
		\caption{${\tt encode}_1(\vc)$}\label{encode1}
		\begin{algorithmic}
			\REQUIRE Data string $\vc$, where $\vc\in\bbracket{q}^{n-1}\times \{1,2,\ldots,q-1\}$.
			\ENSURE $\vx\in \bbracket{q}^{n}$ such that $\vx[\ell-1]\ne \vx[n]$.
			\vspace{0.05in}
			\IF{$\vc[n]\ne\vc[\ell-1]$}
			\STATE $\vx \gets \vc$ 
			\ELSE
			\STATE $\vx \gets $ append $\vc[1,n-1]$ with $0$
			\ENDIF
			\RETURN{$\vx$}
		\end{algorithmic}
	\end{algorithm}
	
	\vspace{2mm}
	
	\begin{algorithm}[!t]
		\caption{${\tt decode}_1(\vx)$}\label{decode1}
		\begin{algorithmic}
			\REQUIRE Codeword $\vx\in \bbracket{q}^n$.
			\ENSURE $\vc\in\{1,2,\ldots,q-1\}^{(n-1)}\times \{1,2,\ldots,q-1\}$.
			\vspace{0.05in}
			
			\IF{$\vx[n]\ne 0$}
			\STATE $\vc\gets\vx$
			\ELSE
			\STATE $\vc\gets$ append $\vx[1,n-1]$  with $\vx[\ell-1]$
			\ENDIF
			\RETURN{$\vc$}
		\end{algorithmic}
	\end{algorithm}

\section{Distinct Profile Vectors from Addressable Codes}\label{sec:addressing}

Borrowing ideas from synchronization, we construct 
a set of words with different profile vectors and prove \eqref{addressing}.
Here, our strategy is to mimic the concept of {\it watermark} and {\it marker codes} 
\cite{Sellers:1962, Kashyap.Neuhoff:1999,Davey.MacKay:2001}, 
where a marker pattern is distributed throughout a codeword.
Due to the unordered nature of the short reads, instead of a single marker pattern, 
we consider a set of patterns. 

More formally, suppose that $0<2a\le \ell\le n$. 
Let $\A=\{\vu_1,\vu_2,\ldots,\vu_M\}\subseteq \bbracket{q}^a$ be a set of $M$ sequences of length $a$.
Elements of $\A$ are called {\em addresses}. 
A word $\vx=\vz_1\vz_2\cdots\vz_M$, where $|\vz_i|=\ell$ for all $1\le i\le M$, 
is said to be {\em $(\A,\ell)$-addressable} 
if the following properties hold.
\begin{enumerate}[(C1)]
\item The prefix of length $a$ of $\vz_i$ is equal to $\vu_i$ for all $1\le i\le M$. In other words, $\vz_i[1,a]=\vu_i$.
\item $\vz_i[j,j+a-1]\notin\A$ for all $1\le i\le M$ and $2\le j\le \ell-a+1$.  
\end{enumerate}
Conditions (C1) and (C2) imply that the address $\vu_i\in\A$ appears exactly once as the prefix of $\vz_i$ and 
does not appear as an $a$-gram of any substring $\vz_j$ with $j\ne i$. 
A code $\C$ is {\em $(\A,\ell)$-addressable} if all words in $\C$ are $(\A,\ell)$-addressable.

Intuitively, 
given an $(\A,\ell)$-addressable word $\vx$, 
we can make use of the addresses in $\A$ to identify the position of each $\ell$-gram in $\vx$ 
and hence, reconstruct $\vx$. We formalize this idea in the following theorem.

\begin{thm}\label{thm:address}
Let $2a\le \ell$ and $\A=\{\vu_1,\vu_2,\ldots, \vu_M\}$ be a set of addresses of length $a$.
Suppose that $\C$ is an $(\A,\ell)$-addressable code.
For distinct words $\vx,\vx'\in\C$, we have $\F(\vx,\ell)\ne \F(\vx',\ell)$.
Therefore, $\vp(\vx,\ell)\ne \vp(\vx',\ell)$ and $P_q(n,\ell)\ge |\C|$.
\end{thm}

\begin{proof}
Let $\vx=\vz_1\vz_2\cdots\vz_M$ and $\vx'=\vz'_1\vz'_2\cdots\vz'_M$ be distinct $(\A,\ell)$-addressable words in $\C$.
Without loss of generality, we assume $\vz_1\ne\vz_1'$. 
Observe that $\vz_1\in \F(\vx,\ell)$. To prove the theorem, it suffices to show that $\vz_1\notin\F(\vx',\ell)$.

Suppose otherwise that  $\vz_1$ appears as an $\ell$-gram in $\vx'$.
Since $\vu_1$ is a prefix of $\vz_1$ with $\vz_1\ne \vz_1'$, by Conditions (C1) and (C2), 
we have that
\[\vx'=\cdots\overbrace{\circ\circ\underbrace{\oplus\oplus\cdots\oplus}_{|\vu_i|=a}\oplus\oplus}^{|\vz_1|=\ell}++\cdots 
\text{ for some }i\ne 1.\]
Here, $\circ$'s and $+$'s represent the $\ell$-grams $\vz_1$ and $\vz'_i$, respectively, and 
$\oplus$'s indicate the symbols that are in the overlap of the two $\ell$-grams.
Since $2a\le \ell$, $\vu_i$ must be in $\vz_1$ as an $a$-gram, contradicting Condition (C2).
\end{proof}

To employ Theorem \ref{thm:address}, we define the following set of addresses,
\begin{equation}\label{eq:addresses}
\A^*\triangleq \left\{(u_1,u_2,\ldots,u_a): \sum_{i=1}^a u_i=0 \bmod q\right\}.
\end{equation}
So, $\A^*$ is a set of $M=q^{a-1}$ addresses and we list the addresses as $\vu_1,\vu_2,\ldots,\vu_M$.
To construct an $(\A^*,\ell)$-addressable code, we consider the encoding map 
${\tt encode}_2:\{1,2,\ldots,q-1\}^{(\ell-a)M}\to \bbracket{q}^{M\ell}$ given in Algorithm \ref{encode}
and define $\C$ to be the image of ${\tt encode}_2$. 
Conversely, we consider the decoding map 
${\tt decode}_2:\C\to \{1,2,\ldots,q-1\}^{(\ell-a)M}$ given in Algorithm \ref{decode}.

\begin{algorithm}[!t]
\caption{${\tt encode}_2(\vc,\A^*)$}\label{encode}
\begin{algorithmic}
\REQUIRE Data string $\vc=\vc_1\vc_2\cdots\vc_M$, 
\STATE ~~~~~~where $\vc_i\in\{1,2,\ldots,q-1\}^{(\ell-a)}$ for $1\le i\le M$,
\STATE ~~~~~~and $\A^*$ is defined by \eqref{eq:addresses}.
\ENSURE $\vx=\vz_1\vz_2\cdots\vz_M\in\bbracket{q}^{M\ell}$, where $\vx$ is $(\A^*,\ell)$-addressable.
\vspace{0.05in}
\FOR{$1\le i\le  M$}
\STATE $\vz_i\gets \vu_i$ ($\vu_i$ has length $a$)
\FOR{$a+1\le j\le  \ell$}
\STATE $z_{\rm bad} \gets -\sum_{s=1}^{a-1}\vz_i[j-s]\bmod q$ 
\STATE~~~~~~~~~(negative of the sum of the last $a-1$ entries modulo $q$)
\STATE $z\gets \vc_i[j-a]$-th element of $(\bbracket{q}\setminus\{z_{\rm bad}\})$
\STATE append $\vz_i$ with $z$
\ENDFOR
\ENDFOR
\RETURN{$\vz_1\vz_2\cdots\vz_M$}
\end{algorithmic}
\end{algorithm}

\vspace{2mm}

\begin{algorithm}[!t]
\caption{${\tt decode}_2(\vz_1\vz_2\cdots\vz_M)$}\label{decode}
\begin{algorithmic}
\REQUIRE Codeword $\vz_1\vz_2\cdots\vz_M \in\C$.
\ENSURE $\vc_1\vc_2\cdots\vc_M\in \{1,2,\ldots,q-1\}^{(\ell-a)M}$.
\vspace{0.05in}

\FOR{$1\le i\le  M$}
\FOR{$a+1\le j\le  \ell$}
\STATE $z_{\rm bad} \gets -\sum_{s=1}^{a-1}\vz_i[j-s]\bmod q$ 
\STATE~~~~~~~~~(negative of the sum of the last $a-1$ entries modulo $q$)
\STATE $\vc_i[j-a]\gets$ the index of $\vz_i[j]$ in $(\bbracket{q}\setminus\{z_{\rm bad}\})$
\ENDFOR
\ENDFOR
\RETURN{$\vc_1\vc_2\cdots\vc_M$}
\end{algorithmic}
\end{algorithm}

\begin{exa} \label{exa:address-encode}For $q=4$ and $a=2$, the address set is $\A^*=\{00,13,22,31\}$ by \eqref{eq:addresses}.
Consider $\ell=5$ and the data string $\vc=(111,123,222,321)$.
Applying Algorithm \ref{encode} to construct $\vz_1$ with $\vc_1=111$, we start with $\vz_1=00$.
Then $z_{\rm bad}=0$ and we choose the first element of $\{1,2,3\}$ to append to $\vz_1$ to get $001$.
In the next iteration, we have $z_{\rm bad}=3$ and append $0$ to $\vz_1$ to get $0010$.
Repeating this, we then obtain $\vz_1=\underline{00}101$. Completing the process for all $i$, we have
\[ \vz_1=\underline{00}101,\
\vz_2=\underline{13}023,\
\vz_3=\underline{22}111,\
\vz_4=\underline{31}210,\]
and so, ${\tt encode}_2(\vc)=(00101,13023,22111,31210)=\vx$.
We check that $\vx$ is indeed $(\A^*,\ell)$-addressable, and
verify that ${\tt decode}_2(\vx)$ in Algorithm \ref{decode} indeed returns the data string $\vc$.
Since there are $3^{12}$ possible data strings, $|\C|=3^{12}\approx 4^{9.51}$.
\end{exa}

Algorithm \ref{encode} bears similarities with a {\em linear feedback shift register} \cite{Golomb:1982}.
The main difference is that we augment our codeword with a symbol that is {\em not equal} 
to the value defined by the linear equation. This then guarantees that we have no $a$-grams belonging to $\A^*$.
More formally, we have the following proposition. 

\begin{prop}\label{prop:constr}Consider the maps ${\tt encode}_2$ and ${\tt decode}_2$  defined by Algorithms \ref{encode} and \ref{decode}
and the code $\C$. Then $\C$ is an $(\A^*,\ell)$-addressable code and 
${\tt decode}_2\circ{\tt encode}_2(\vc)=\vc$ for all $\vc \in \{1,2,\ldots,q-1\}^{(\ell-a)M}$.
Hence, ${\tt encode}_2$ is injective and $|\C|=(q-1)^{M(\ell-a)}$.
Furthermore, ${\tt decode}_2$ and ${\tt encode}_2$ computes their respective strings in $O(qM\ell)$ time.
\end{prop}

Since $M=q^{a-1}$, Theorem \ref{thm:address} and Proposition \ref{prop:constr} then yield \eqref{addressing}
for $n=q^{a-1}\ell$ and $2a\le \ell$.
In other words,  for $n=q^{a-1}\ell$ and $2a\le \ell$, we have \[R_q(n,\ell)\ge \left(1-\frac{a}{\ell}\right) \log_q (q-1).\]

We now modify our construction to derive addressable codes for all values of 
$n\le q^{\floor{\ell/2}-1}\ell$. 
Suppose that $m=\floor{n/\ell}$. Choose $a=\ceiling{\log_q m}+1$ so that $m\le q^{a-1}$.
Use a subset $\B^*$ of $\A^*$ of size $m$ for the address set.
A straightforward modification of Algorithm \ref{encode} 
then yields $(\B^*,\ell)$-addressable words of the form
\[\vu_1\underbrace{\circ\circ\cdots\circ}_{\ell-a}
\vu_2\underbrace{\circ\circ\cdots\circ}_{\ell-a}\cdots
\vu_m\underbrace{\circ\circ\cdots\circ}_{\ell-a}
\underbrace{00 \cdots 0}_{n-m\ell}\, .\]
The size of this $(\B^*,\ell)$-addressable code can be computed to be $(q-1)^{m(\ell-a)}$.
We obtain the following corollary.

\begin{cor}\label{cor:address}
For $n\le q^{\floor{\ell/2}-1}\ell$, suppose that $n=m\ell+t$ with $0\le t<\ell$. 
Set $a=\ceiling{\log_q m}+1$ so that $m\le q^{a-1}$ and $2a\le \ell$.
Then $P_q(n,\ell)\ge (q-1)^{m(\ell-a)}$, or,
\[
R_q(n,\ell)\ge \left(\frac{m(\ell-a)}{n}\right)\log_q(q-1)\approx \left(1-\frac{a}{\ell}\right)\log_q(q-1).
\]
\end{cor}

\begin{exa}
Setting $q=4$, $a=5$, and $\ell=100$ in \eqref{addressing} yields $P_4(25600, 100)\ge 3^{24320}\approx 4^{19273}$.
In other words, $R_4(25600, 100)\ge 0.753$. 
Applying Corollary \ref{cor:address} and 
varying $a\in\{2,3,4,5\}$, we have 
\[R_4(100m, 100)\ge 
\begin{cases}
0.776, &\mbox{ for $1< m\le 4$;}\\
0.768, &\mbox{ for $4< m \le 16$;}\\
0.760, &\mbox{ for $16< m\le 64$;}\\
0.752, &\mbox{ for $64< m\le 256$;}
\end{cases}\]
We improve the above lower bound for $R_4(n,100)$ in Section \ref{sec:debruijn}.
Nevertheless, this family of profile vectors has efficient encoding and decoding algorithms. 
We demonstrate a simple assembly algorithm in the next subsection.
\end{exa}

\subsection{Assembly of $(\A,\ell)$-Addressable Words in the Presence of Coverage Errors}


Let $\A$ be a set of addresses, $\C$ be a set of $(\A,\ell)$-addressable words, 
and $\vx\in \C$. 
In this subsection, we present an algorithm that takes the set of reads provided by $\vp(\vx,\ell)$
and correctly assembles $\vx$. 
We also observe that correct assembly is possible even if some reads are lost.

We use the formal definition of the DNA storage channel given by Kiah \etal \cite{Kiah.etal:2015a}
 and reproduce here the notion of {\em coverage errors}.
Suppose that the data of interest is encoded by a vector  $\vx\in \llbracket q\rrbracket^n$ and 
let $\vhp(\vx)$ be the output profile vector of the channel.
Coverage errors occur when not all $\ell$-grams are observed during fragmentation and subsequently sequenced.
For example, suppose that $\vx=10001$ from Example \ref{exa:simple}, and that $\vhp(\vx)$ is the channel output $2$-gram profile vector.
The coverage loss of one $2$-gram results in the count of $00$ in $\vhp(\vx)$ to be one instead of two.

We then define $\B(\vx,e)$ to be the set of all possible output profile vectors arising from at most $e$ coverage errors 
with input vector $\vx$. Then for a code $\C\subseteq \bbracket{q}^n$,
a map $\Phi: \nonneg^{q^\ell}\to \C \cup \{{\rm Fail}\}$ is an {\em assembly algorithm for $\C$ that corrects $e$ coverage errors} if for all $\vx\in \C$,
$\Phi(\vhp)=\vx \mbox{ for all }\vhp\in \B(\vx,e)$.

%
%
%
%
Let $\vhp$ represent  a (possibly incomplete) set of reads obtained from $\vx$.
Our broad strategy of assembly is as follows.
For each read $\vr\in \vhp$, we first attempt to {\em align} $\vr$ by guessing the index $j$ such that  $\vr=\vx[j,j+\ell-1]$.
After which, we ensure that all symbols in $\vx$ are {\em covered} by some correctly aligned read,
so that the entire word $\vx$ can be reconstructed.

We now describe in detail the alignment step.
Let $\vx$ be an $(\A,\ell)$-addressable word. Recall that $\vx$ can be written as
\[
\vx= \underbrace{\overbrace{x_1\cdots x_a}^{\vu_1}\cdots x_{\ell}}_{\vz_1}\cdots
\underbrace{\overbrace{x_{(i-1)\ell+1}\cdots x_{(i-1)\ell+a}}^{\vu_i}\cdots x_{i\ell}}_{\vz_i}\cdots
\underbrace{\overbrace{x_{(M-1)\ell+1}\cdots x_{(M-1)\ell+a}}^{\vu_M}\cdots x_{M\ell}}_{\vz_M} \, ,
\]
so that $\vx[(i-1)\ell+1,(i-1)\ell+a]=\vu_i$ for all $1\le i\le M$.

Let $\vr$ be a read obtained from $\vx$ and 
we say that $\vr$ is {\em correctly aligned at $j$} for some $1\le j\le n-\ell+1$
if $\vx[j,j+\ell-1]=\vr$.
To align the read $\vr$, our plan is to look for the address 
that occurs last in $\vr$, say $\vu_i$,
and match it to the corresponding index $(i-1)\ell+1$.
Specifically, suppose that $\vr$ is a read that contains an address as a substring. 
Let $t$ be the largest index such that $\vr[t,t+a-1]=\vu_i\in \A$. 
Then we align $\vr$ in a way such that $\vr[t,t+a-1]$ matches the address $\vu_i$ (see Figure \ref{fig:align}(a)).

However, not all reads can be aligned correctly. 
To remedy the situation, we define a special type of read that can be proven to be always aligned correctly.




\begin{defn}\label{def:typei}

Let $\vr\in\bbracket{q}^\ell$ be an $\ell$-gram or read.
Define $L(\vr)$ to be the the largest index $t$ of $\vr$ such that $\vr[t,t+a-1] \in \A$. 
If such a $t$ does not exist, then we set $L(\vr)=\infty$.
We say that $\vr$ is a {\em Type I read} if $L(\vr)\le \ell-2a+2$.
\end{defn}

\begin{exa} \label{exa:typei}In Definition \ref{def:typei}, observe that we can characterize a Type I read 
without knowing $\vx$. This is illustrated in the examples below.

\begin{enumerate}
\item Consider the set of parameters in Example \ref{exa:address-encode} 
with $q=4$, $a=2$, $\A^*=\{00,13,22,31\}$, and $\ell=5$.
Consider further the $(\A^*,\ell)$-addressable word  $\vx=(00101, 13023, 22111, 31210)$.
The 5-gram $01130$ is a Type I read, since $L(01\underline{13}0)=3\le 3$
On the other hand, the 5-grams $30232$ and $21113$ are {\em not} Type I reads,
since $L(30232)=\infty$ and $L(211\underline{13})=4> 3$.


\item Consider $q=4$, $a=3$, $\ell=8$, and $\A^*$ given by \eqref{eq:addresses}.
The 8-gram $101\underline{220}01$ is a Type I read, while the 8-grams $32122133$ and $2130\underline{130}3$ are {\em not} Type I reads. Note that $L(21301303)=5$, since $130$ is the last $3$-gram even though 
$301$ and $013$ also occur as $3$-grams of the read $21301303$.
\end{enumerate}

\end{exa}
 

Now, our alignment method may fail when a read is {\em not} a Type I read.
For example, if we consider $\vx$ in Example \ref{exa:typei}(i) and the read $\vr=211\underline{13}$, 
our method aligns $\vr$ to index $3$, which is incorrect. 
In contrast, if a read is of Type I, we prove that the alignment is always correct.

\begin{figure}[!t]
\begin{enumerate}[(a)]
\item \hfill
\begin{center}
\includegraphics[width=16cm]{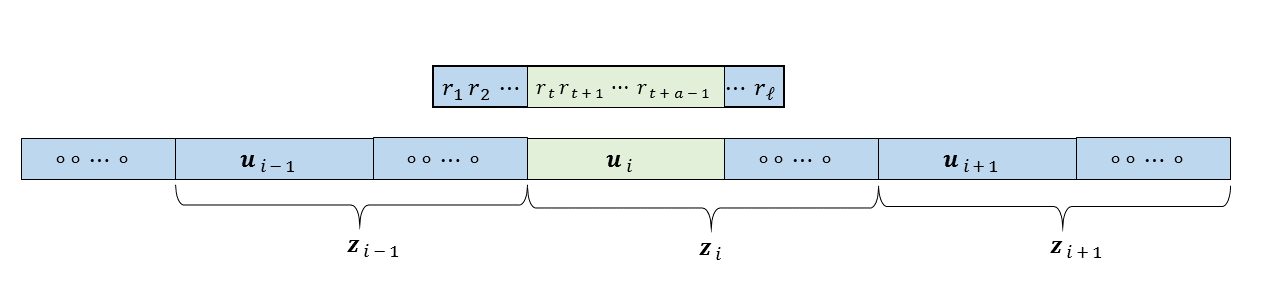}
\end{center}
\item \hfill
\begin{center}
\includegraphics[width=16cm]{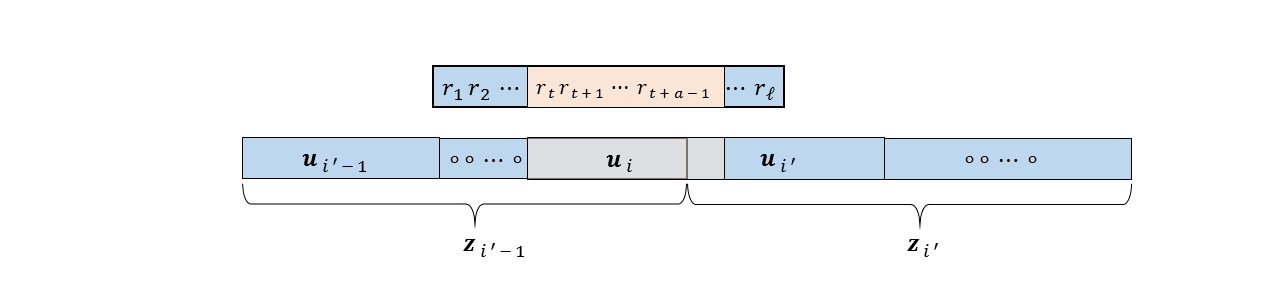}
\end{center}
\end{enumerate}

\caption{Possible ways of obtaining a read containing the address $\vu_i$. 
(a) The string $\vu_i$ is the prefix of the $i$th component $\vz_i$ of $\vx$. 
(b) The string $\vu_i$ has an overlap with the prefix of the $i'$th component $\vz_{i'}$ of $\vx$.   }
\label{fig:align}
\end{figure}

\begin{lem}\label{lem:align} Let $\vr$ be a Type I read of an $(\A,\ell)$-addressable word $\vx$. 
If $L(\vr)=t$ and $\vr[t,t+a-1]=\vu_i$, 
then $\vr$ is correctly aligned at $j=(i-1)\ell-t+2$.
\end{lem}

\begin{proof}
It suffices to show that if $\vx[j',j'+\ell-1]=\vr$ for some index $j'$, then $j'$ is necessarily equal to $j$.

Now, $\vr[t,t+a-1]=\vu_i$, or equivalently, $\vx[j'+t-1,j'+t+a-2]=\vu_i$. 
From the definition of an $(\A,\ell)$-addressable word,
$\vu_i$ cannot be properly contained in any of the substrings $\vz_1,\vz_2,\ldots,\vz_M$, with the exception of $\vz_i$, where $\vu_i$ is the prefix. 
Therefore, we only have the following two possibilities for the index of $j'+t-1$ (see Figure \ref{fig:align}).
\begin{enumerate}[(a)]
\item $\vu_i$ is a prefix of $\vz_i$. Then $j'+t-1=(i-1)\ell+1$, implying that $j'=(i-1)\ell-t+2$.
\item The substring $\vr[t,t+a-1]=\vx[j'+t-1,j'+t+a-2]=\vu_i$ has an overlap with some address $\vu_{i'}\in \A$.
Specifically, $j'+t-1 \pmod \ell \in \{0,\ell-a+2,\ell-a+3,\ldots,\ell-1\}$. 
Since $\vr$ is a Type I read, we have that $t\le \ell-2a+2$. 
Hence, there are at least $a-1$ symbols in $\vr$ after the substring $\vu_i$.
Since $\vu_i$ has some overlap with $\vu_{i'}$, the substring $\vu_{i'}$ is necessarily contained in $\vr$. 
This contradicts the maximality of $t$.
\end{enumerate}
In conclusion, the only possibility is $j'=j$.
\end{proof}

It remains to show that all symbols in $\vx$ is covered by a Type I read.
More formally, let $1\le i\le n$ and we say that a read $\vr=\vx[j,j+\ell-1]$ {\em covers the symbol $\vx[i]$} if $j\le i\le j+\ell-1$.
The following lemma characterizes Type I reads.

\begin{lem}
Let $1\le j\le n-\ell+1$.
The $\ell$-gram $\vx[j,j+\ell-1]$ is a Type I read if and only if 
$j\pmod \ell \not\in \{2,3,\ldots, 2a-1\}$.
\end{lem}

The next corollary then follows from a straightforward computation.

\begin{cor}\label{cor:cover}
Let $\vx$ be a $q$-ary word of length $n$.
Each symbol in $\vx$ is covered by at least one Type I read.
In particular, if $\ell\le i\le n-\ell$,  then $\vx[i]$ is covered by exactly $\ell-2a+2$ Type I reads. 
\end{cor}

Lemma \ref{lem:align} and Corollary \ref{cor:cover} imply the correctness of our assembly algorithm for $\C$,
provided that no reads are lost. 
In order for the assembly algorithm to correct coverage errors,
we simply fix the values of the first $\ell$ and the last $\ell$ symbols of all codewords.
Then each of the other $n-2\ell$ symbols is covered by exactly $\ell-2a+2$ Type I reads.
Therefore, if at most $\ell-2a+1$ reads are lost, we are guaranteed that all of these $n-2\ell$ symbols 
are covered by at least one  Type I read.
More formally, we have the following theorem.

\begin{thm}\label{thm:address-assemble}
Let $\A=\{\vu_1,\vu_2, \ldots, \vu_M\}$ be a set of addresses. Suppose that
 $\vz_1^*, \vz_M^*\in \bbracket{q}^\ell$ obey the following properties:
 \begin{enumerate}[(1)]
\item $\vz_1^*[1,a]=\vu_1$ and $\vz_M^*[1,a]=\vu_M$;
\item $\vz_1^*[j,j+a-1]^*\notin\A$ and $\vz_M^*[j,j+a-1]\notin\A$ for all $2\le j\le \ell-a+1$.  
\end{enumerate}

\noindent Let $\C^*$ be a set of $(\A,\ell)$-addressable words 
such that $\vx[1,\ell]=\vz_1^*$ and $\vx[M\ell-\ell+1,M\ell]=\vz_M^*$ for all $\vx\in\C^*$.
Then there exists an assembly algorithm for $\C^*$ that corrects $\ell-2a+1$ coverage errors.
\end{thm}

\begin{rem}With suitable modifications to code $\C^*$ given in Theorem \ref{thm:address-assemble},
we obtain addressable codes that correct more than $\ell-2a+1$ coverage errors. 
We sketch the main ideas as they follow from the well-known concept of code concatenation \cite{Forney:1966} 
and \cite[Section 5.5]{Huffman:2010}.
More concretely, we assume two codes: 
an {\em outer} code $\D_{\rm out}$ of length $M-2$ over an alphabet $\Sigma$, and 
an {\em inner} code $\D_{\rm in}$ of length $\ell-a$ over the alphabet $\bbracket{q}$, 
that satisfies the following conditions.
\begin{enumerate}[(D1)]
\item $|\Sigma|\le |\D_{\rm in}|$, and so, there exists an injective map $\chi:\Sigma\to \D_{\rm in}$.
\item Recall the definition of $\vz_1^*$ and $\vz_M^*$ in Theorem \ref{thm:address-assemble}. 
Define $\D^*$ to be the $q$-ary code of length $M\ell$ given by 
\[\D^*\triangleq \left\{\vz_1^*\vz_2\cdots\vz_{M-1}\vz_M^*: \vz_j=\vu_j\chi(d_j), 2\le j\le M-1, 
(d_2,\ldots, d_{M-1})\in \D_{\rm out}\right\}.\]
We require $\D^*$ to be $(\A,\ell)$-addressable.
\end{enumerate}
Therefore, if the outer code $\D_{\rm out}$ and inner code $\D_{\rm in}$ are able to correct up to $(D-1)$ and $(d-1)$ erasures, respectively, 
then the assembly algorithm for $\D^*$ is able to correct $dD(\ell-2a+2)-1$ coverage errors. 
\end{rem}


\section{Distinct Profile Vectors from Partial de Bruijn Sequences}
\label{sec:debruijn}

We borrow classical results on de Bruijn sequences and certain results from Maurer~\cite{Maurer92}
to provide detailed proofs for \eqref{debruijn} and Corollary \ref{cor:rates}. 
We first define partial $\ell$-de Bruijn sequences.

\begin{defn}
A $q$-ary word $\vx$ is a {\em partial $\ell$-de Bruijn sequence} if every $\ell$-gram of $\vx$ appears at most once.
In other words, $p(\vx,\vz)\le 1$.
A partial $\ell$-de Bruijn sequence is {\em complete $\ell$-de Bruijn} if every $\ell$-gram of $\vx$ appears exactly once.
\end{defn}

By definition, a complete $\ell$-de Bruijn sequence has length $q^\ell+\ell-1$.
The number of distinct complete $\ell$-de Bruijn sequences was explicitly determined by van Arden-Ehrenfest
and de Bruijn \cite{AEB87} as a special case of a more general result on the number of trees 
in certain graphs. They built upon a previous work done by Tutte and Smith (see the note at the end of~\cite{AEB87}). Their combined effort led to a formula of counting the number of trees in a directed graph as the value of a certain determinant. The formula is now known as BEST Theorem. The acronym refers to the four surnames, namely de {\bf B}ruijn, {\bf E}hrenfest, {\bf S}mith, and {\bf T}utte.
\begin{thm}[BEST\cite{AEB87}]\label{thm:best}
The number of distinct complete $\ell$-de Bruijn words%
 is $(q!)^{q^{\ell-1}}$.
\end{thm}

\begin{rem}
Theorem \ref{thm:best} is usually stated in terms of Eulerian circuits in a de Bruijn graph of order $\ell$. 
We refer the interested reader to van Arden-Ehrenfest and de Bruijn ~\cite{AEB87} for the formal definitions.
Specifically, the number of Eulerian circuits is known to be $(q!)^{q^{\ell-1}}/q^\ell$. 
Consider a circuit represented by the $q$-ary word $\vy$.
To obtain $q^\ell$ complete $\ell$-de Bruijn sequences, we simply consider the $q^\ell$ rotations of $\vy$
and append to each rotation $\vy'$ its prefix of length $\ell-1$.  
\end{rem}

Partial de Bruijn sequences are of deep and sustained interest in graph theory, combinatorics, and cryptography. 
In the first two, their inherent structures are the focus of attention, 
while, in cryptography, the interest is mainly on the case of $q=2$ to generate random looking sequences to be used as keystream in additive stream ciphers~\cite[Sect. 6.3]{Menezes:1996}. 
Mauer established the following enumeration result.

\begin{thm}[{\cite[Theorem 3.1]{Maurer92}}]\label{thm:maurer}
The number of partial $\ell$-de Bruijn sequences is larger than $L_q(n) - \binom{n}{2} {q^{n-\ell}}/{n}$.
\end{thm}

Recall that $L_q(n)$ is the number of $q$-ary Lyndon words of length $n$ given by \eqref{eq:lyndon}.
Next, we make the connection to profile vectors based on the following observation of Ukkonen.

\begin{lem}[{\cite[Theorem 2.2]{ukkonen1992}}]\label{lem:ukk}
Let $\vx$ be a $q$-ary sequence of length $n$ such that every $(\ell-1)$-gram appears at most once. 
If $\vy$ is a $q$-ary sequence such that $\vp(\vx,\ell)=\vp(\vx',\ell)$, then $\vx=\vx'$.
\end{lem}

It follows from Lemma \ref{lem:ukk} that two distinct partial $(\ell-1)$-de Bruijn sequences have distinct $\ell$-gram profile vectors. Therefore, the number of distinct partial $(\ell-1)$-de Bruijn sequences is a lower bound for $P_q(n,\ell)$.
Hence, we establish \eqref{debruijn} and also the following corollary to BEST Theorem.

\begin{cor} Let $n=q^\ell+\ell-1$. Then $P_q(n,\ell)\ge (q!)^{q^{\ell-1}}$. Hence, $\alpha(n,q)\ge \frac1q\log_q (q!)$.
\end{cor}

Next, suppose that $\ell\ge 2\log_q n +2$ with $n\ge 8$. We conduct an analysis similar to Maurer's 
to complete the proof of Theorem \ref{thm:main}(iii). 

The following two inequalities can be established.
\[ \binom{n}{2}q^{n-\ell+1}<\frac{q^{n-1}}2, \mbox{ since } q^{\ell-2}\ge n^2, \mbox{ and }
\sum_{\substack{d<n\\d|n}} \mu\left(\frac nd\right) q^{d}<\frac n2 q^{n/2}\le \frac12 q^{n-1}\mbox{ for }
q\ge 2, n\ge 8.\]
Therefore, 
\[ nL_q(n)- \binom{n}{2} {q^{n-\ell+1}} \ge q^n-\frac12 q^{n-1}-\frac12 q^{n-1}\ge q^{n-1}.\]
Hence, $P_q(n,\ell)>q^{n-1}/n$.

\vspace{2mm}

To end this section, we prove Corollary \ref{cor:rates}.
First, if $n\le q^{\ell/2-1}$, or $\ell\ge 2\log_q n +2$, then $R_q(n,\ell)\ge (n-1-\log_qn)/n$.
Therefore, after taking limits, we have $\alpha(n,q)=1$, proving Corollary \ref{cor:rates}(i).

%

Next from \eqref{upper}, we have that, for $n\ge 2$, $\ell\ge 2$,
\[P_q(n,\ell)\le \binom{n-\ell+q^\ell}{q^\ell-1}=\prod_{j=1}^{q^\ell-1}\frac{n-\ell+j+1}{j}\le n^{q^\ell-1}.\]
Then $R_q(n,\ell)\le (q^\ell-1)(\log_q n) /n$.
Fix $\epsilon>0$.
If $n\ge q^{(1+\epsilon)\ell}$, then $q^\ell\le n^{1/(1+\epsilon)}$. 
So,
\[R_q(n,\ell)\le \frac{(q^\ell-1)\log_q n}{n}\le  \frac{n^{1/(1+\epsilon)}\log_q n}{n}=  \frac{\log_q n}{n^{\epsilon/(1+\epsilon)}}.\]
After taking limits, we have $\alpha(n,q)=0$, proving Corollary \ref{cor:rates}(ii).

\section{Conclusion}

We provided exact values and lower bounds  
for the number of profile vectors given moderate values of $q$, $\ell$, and $n$.
Surprisingly, for fixed $q\ge 2$ and moderately large values $n\le q^{\ell/2-1}$,
the number of profile vectors is at least $q^{\kappa n}$ with $\kappa$ very close to 1.
In other words, for practical values of read and word lengths, 
we are able to obtain a set of distinct profile vectors with rates close to one.
In addition to enumeration results, we propose a set of linear-time encoding and decoding algorithms 
for each of two particular families of profile vectors..

In our future work, we want to provide sharper estimates on 
the asymptotic rate of profile vectors $\alpha(n,q)$ (see \eqref{def:cap})
when $q^\ell/2 \le n\le q^\ell$ 
and to examine the number of profile vectors with specific $\ell$-gram constraints {\em a la} Kiah \etal \cite{Kiah.etal:2015a}.

%
%


\bibliographystyle{IEEEtran}
\bibliography{mybibliography.bib}

\end{document}